\documentclass[12pt,a4paper,twoside,reqno]{amsart}

\usepackage{amsmath}
\usepackage{amsfonts}
\usepackage{amsthm}
\usepackage{fullpage}
\usepackage{paralist}
\usepackage{verbatim}
\usepackage{graphicx}

\usepackage{enumitem}
\newlist{AG}{enumerate}{1}
\setlist[AG]{label=(AG\arabic*)}

\usepackage{etoolbox}
\patchcmd{\subsection}{\bfseries}{\itshape}{}{}
\patchcmd{\subsection}{-.5em}{.5em}{}{}

\theoremstyle{plain}
   \newtheorem{claim}{Claim}[section]

   \newtheorem{lem}[claim]{Lemma}
   \newtheorem{obs}[claim]{Observation}
   \newtheorem{prop}[claim]{Proposition}
   \newtheorem{thm}[claim]{Theorem}
\theoremstyle{definition}
   
   \newtheorem{eg}[claim]{Example}
\theoremstyle{remark}

\DeclareMathOperator*{\Count}{Count}
\DeclareMathOperator*{\CountPS}{CountWB}
\DeclareMathOperator*{\GrpContribution}{GrpContrib}
\DeclareMathOperator*{\MatContribution}{MatContrib}
\DeclareMathOperator*{\ModCount}{ModCount}
\DeclareMathOperator*{\Multiple}{Mult}

\begin{document}

\title{On the star-height of subword counting languages and their relationship to Rees zero-matrix semigroups}

\author{Tom Bourne}
\address{School of Mathematics and Statistics, University of St Andrews, St Andrews, Scotland, U.K.}
\email{tom.bourne@st-andrews.ac.uk}

\author{Nik Ru\v{s}kuc}
\address{School of Mathematics and Statistics, University of St Andrews, St Andrews, Scotland, U.K.}
\email{nik.ruskuc@st-andrews.ac.uk}

\date{}

\begin{abstract}
Given a word \( w \) over a finite alphabet, we consider, in three special cases, the generalised star-height of the languages in which \( w \) occurs as a contiguous subword (factor) an exact number of times and of the languages in which \( w \) occurs as a contiguous subword modulo a fixed number, and prove that in each case it is at most one. We use these combinatorial results to show that any language recognised by a Rees (zero-)matrix semigroup over an abelian group is of generalised star-height at most one.
\end{abstract}

\keywords{Regular language, star-height, subword, Rees matrix semigroup}

\maketitle

\section{Introduction and preliminaries}
\label{sec:intro}\label{sec:prelims}

The generalised star-height problem, which asks whether or not there exists an algorithm to compute the generalised star-height of a regular language, is a long-standing problem in the field of formal language theory.
In particular, it is not yet known whether there exist languages of generalised star-height greater than one; see
\cite[Section I.6.4]{sak09} and \cite{pin92}.
The aim of the present paper is to present some new contributions concerning this problem. In Section~\ref{sec:countingSubwords}, we take a combinatorial approach and find the generalised star-height of languages where a fixed word \( w \) appears as a contiguous subword of the words in our language precisely \( k \) times or \( k \) modulo \( n \) times.
In Section~\ref{sec:ReesMatrix}, we apply these results to prove that languages recognised by Rees (zero-)matrix semigroups over abelian groups are of generalised star-height at most one.

An \emph{alphabet} $A$ is a finite, non-empty set; its elements are \emph{letters}.
A finite sequence of letters is a \emph{word (over \( A \))}. The \emph{length} of a word \( w \), denoted by \( | w | \), is the total number of letters appearing in \( w \). The \emph{empty word}, denoted by \( \varepsilon \), is the unique word of length zero. The set of all words over \( A \) is denoted by \( A^{\ast} \), and the set of all non-empty words over \( A \) is denoted by \( A^{+} \). A \emph{semigroup} (respectively, \emph{monoid}) \emph{language} is a subset of \( A^{+} \) (respectively, \( A^{\ast} \)).

Given an alphabet \( A \), we define the empty set, the empty word, and each of the letters in \( A \) to be \emph{basic regular expressions}. Using these, we recursively define new \emph{regular expressions} by using the (finite) union, concatenation product, (Kleene) star and complement operations; that is, if \( E \) and \( F \) are regular expressions then so too are \( E \cup F \), \( EF \), \( E^{*} \) and \( E^{c} \). 
A language is \emph{regular} if it can be represented by a regular expression.

The (generalised) \emph{star-height} of a regular expression \( E \), denoted by \( h(E) \), is defined recursively as follows: for the basic regular expressions, \( h(\emptyset) = h(\varepsilon) = h(a) = 0 \), where \( a \) is a letter from \( A \); for union and product, \( h(E \cup F)= h(EF) = \mathrm{max}\{h(E), h(F)\} \); for the star operation, \( h(E^{\ast}) = h(E) + 1 \); and for complementation, \( h(E^c) = h(E) \). The (generalised) \emph{star-height} of a language \( L \), denoted by \( h(L) \), is
\begin{align*}
h(L) = \mathrm{min}\{ h(E) \mid E \text{ is a regular expression representing } L \}.
\end{align*}
Note that we can use De Morgan's laws to express intersection and set difference, and that
\begin{align*}
h(E \cap F) = h(E \setminus F) = \mathrm{max}\{h(E), h(F)\}.
\end{align*}

It is well known that the class of regular languages remains unchanged if complementation is removed from the list of allowed operations. One can define the notion of (restricted) star-height with respect to this signature.
In this context, the star-height problem has been solved: there exist languages of arbitrary (restricted) star-height \cite{eggan63}, and the (restricted) star-height of a language is algorithmically computable \cite{hashiguchi83}.

From this point on, the phrase ``star-height'' will always refer to generalised star-height.

The following simple observation, which allows `removal' of stars, will be used throughout the paper:

\begin{obs}[\cite{pin89}]
\label{obs:subsetOfAlphabet}
For any alphabet \( A \) and any subset \( B \) of \( A \) we have 
\[
A^\ast= \emptyset^c \qquad \text{and} \qquad  B^\ast = A^\ast\setminus (A^\ast (A\setminus B) A^\ast).
\]
Hence,
\[
h(A^\ast)=h(B^{\ast}) = 0.
\]
\end{obs}

Let \( u, w \) and \( x \) be elements of \( A^{\ast} \). If \( v = uwx \) then \( u \) is a \emph{prefix} of \( v \), \( w \) is a \emph{contiguous subword} (or \emph{factor}) of \( v \), and \( x \) is a \emph{suffix} of \( v \). Throughout this paper, the phrase ``subword'' will always mean contiguous subword. A prefix of a word that is also a suffix of that word is a \emph{border}, and the proper border of greatest length is said to be \emph{maximal}.

For every word \( w \) in \( A^{+} \) and every word \( v \) in \( A^{\ast} \), we denote the number of times that \( w \) appears as a subword of \( v \) by \( |v|_{w} \). When \( w \) is a letter, say \( w = a \), the notation \( |v|_{a} \) coincides with its usual meaning; that is, the number of times the letter \( a \) appears in the word \( v \). For every word \( w \) in \( A^{+} \) and every non-negative integer \( k \), we define the language \( \Count(w, k) \) by
\[
\Count(w, k) = \{v \in A^{\ast} \mid |v|_{w} = k\};
\]
that is, the set of words \( v \) over \( A \) such that \( w \) appears as a subword of \( v \) precisely \( k \) times. As such, we regard \( \Count(w, 0) \) as the set of all words that do not feature \( w \) as a subword. From this characterisation, we note that
\begin{align}
\label{eq:Count0}
v \in \Count(w, 0) \Leftrightarrow v \in A^{\ast} \setminus A^{\ast} w A^{\ast} \Leftrightarrow v \in \left(A^{\ast} w A^{\ast}\right)^{c} \Leftrightarrow v \in \left( \emptyset^{c} w \emptyset^{c} \right)^{c},
\end{align}
where the the final equivalence follows by Observation~\ref{obs:subsetOfAlphabet}. Thus, for a fixed word \( w \), the language \( \Count(w, 0) \) is representable by a star-free expression and is therefore of star-height zero.

In a similar manner, for every word \( w \) in \( A^{+} \),  every integer \( n \) greater than or equal to \( 2 \)  and every non-negative integer \( k \) with \( 0 \leq k < n \), we define the language \( \ModCount(w, k, n) \) by
\[
\ModCount(w, k, n) = \{v \in A^{\ast} \mid |v|_{w} \equiv k \pmod n\};
\]
that is, the set of words \( v \) over \( A \) such that \( w \) appears as a subword of \( v \) precisely \( k \) modulo \( n \) times.

It should be noted that the languages \( \Count(w, k) \) and \( \ModCount(w, k, n) \) are regular. 
This can be proved directly; for example, by building a finite state automaton accepting the language and appealing to Kleene's Theorem (see, for example, \cite[Theorem I.2.3]{sak09}). For the languages under consideration in this paper, regularity also follows from the proofs in Section~\ref{sec:countingSubwords}.

In Section~\ref{sec:countingSubwords} we prove the following result:

\begin{prop}
Let \( A \) be an alphabet. For any word \( w \) in \( A^{+} \) with \( |w| \leq 3 \), the language \( \Count(w, k) \) is of star-height zero, and the language \( \ModCount(w, k, n) \) is of star-height at most one.
\end{prop}

In Section~\ref{sec:ReesMatrix} we are interested in languages recognised by Rees (zero-)matrix semigroups over abelian groups. A language \( L \subseteq A^{+} \) is \emph{recognised} by a semigroup \( S \) if there exists a semigroup morphism \( \varphi : A^{+} \to S \) and a subset \( X \) of \( S \) such that \( L = X \varphi^{-1} \). Again by Kleene's Theorem, a language is recognisable by a finite semigroup if and only if it is regular. We then prove:

\begin{thm}
A language recognised by a Rees (zero-)matrix semigroup over an abelian group is of star-height at most one.
\end{thm}

In order to prove this, we combine the results of Section~\ref{sec:countingSubwords} and a general result on Rees zero-matrix semigroups over semigroups with the following known results:

\begin{AG}
\item
\label{lem:commGrpIffBool}
A language \( L \) is recognised by a finite abelian group if and only if \( L \) is a boolean combination of languages of the form \( \ModCount(a, k, n) \), where \( a \) is a letter from an alphabet \( A \); see, for example, \cite[Corollary 2.3.12]{pin86}.
\item
\label{prop:commGrpHeightOne}
A language recognised by a finite abelian group is of star-height at most one; \cite{henneman71}.
\end{AG}

\section{Counting subwords}
\label{sec:countingSubwords}

Throughout this section we will consistently make use of the notation \( \Count(w, k) \) and \( \ModCount(w, k, n) \) as introduced in Section~\ref{sec:prelims}. We split our analysis into the following three special cases:
\begin{enumerate}[label=(\arabic*), widest=(2), leftmargin=10mm]
\item Counting subwords over a unary alphabet;
\item Counting subwords with maximal border \( \varepsilon \) over a non-unary alphabet;
\item Counting subwords that are a power of a letter over a non-unary alphabet.
\end{enumerate}

Case (1) is simple, but we consider it for the sake of establishing some equalities that will be useful subsequently.
The substantive difference between cases (2) and (3) is that in (3) the letters of \( w \) may appear as components of multiple subwords. For example, if we are counting the number of occurrences of \( aa \) and we encounter the expression \( aaa \) then we have two occurrences of \( aa \) and the central \( a \) belongs to both. This is not a problem in the second case as having maximal border \( \varepsilon \) ensures that
occurrences of \( w\) do not overlap.

\subsection{Case 1: a unary alphabet}
\label{sec:countingUnaryAlphabet}

Let \( A=\{a\} \) be a unary alphabet. A language \( L \) over \( A \) is regular if and only if \( L \) is of the form \( X \cup Y (a^{r})^{\ast} \), where \( X \) and \( Y \) are finite sets and \( r \) is an integer greater than or equal to \( 0 \); see \cite[Exercise II.2.4]{sak09}. Thus, every language over a unary alphabet is of star-height at most one. However, we want to find expressions of minimal star-height for the languages \( \Count(a^{r}, k) \) and \( \ModCount(a^{r}, k, n) \), where \( r \) is a natural number, to be used later in the non-unary cases.

We begin with finding an expression for \( \Count(a^{r}, k) \). If we consider an arbitrary word \( a^s \) then each \( a \) appearing in it is the start of an occurrence of \( a^r \), except for the final \( r-1 \) letters. It immediately follows that
\begin{eqnarray}
\label{eq:Count(a^r,0)}
&&\Count(a^{r}, 0) = \varepsilon \cup a \cup \dots \cup a^{r-1},\\
\label{eq:Count(a^r,k)}
&&\Count(a^{r}, k) = a^{r+k-1}\ (k>0).
\end{eqnarray}

Next we find an expression for \( \ModCount(a^{r}, k, n) \).  The approach taken is to first count \( k \) occurrences of the subword \( a^{r} \) and then repeat in multiples of \( n \). Recalling the expression for \( \Count(a^{r}, k) \) in \eqref{eq:Count(a^r,k)} we obtain
\[
\ModCount(a^{r}, k, n) = a^{r+k-1}(a^{n})^{\ast}.
\]
An expression for the remaining language, namely \( \ModCount(a^{r}, 0, n) \), 
is obtained by using similar reasoning, but keeping in mind the special nature of \( \Count(a^{r}, 0) \) as in 
\eqref{eq:Count(a^r,0)}; it yields
\[
\ModCount(a^{r}, 0, n) = \varepsilon \cup a \cup \dots \cup a^{r-1} \cup a^{r+n-1}(a^{n})^{\ast}.
\]


A combination of the above constitutes a proof for the following lemma:

\begin{lem}
\label{lem:unaryAlphabetHeights}
Let \( A = \{a\} \) be a unary alphabet. For every natural number \( r \), the language \( \Count(a^{r}, k) \) is of star-height zero, and the language \( \ModCount(a^{r}, k, n) \) is of star-height at most one. \qed
\end{lem}

\subsection{Case 2: a non-unary alphabet and maximal border \( \varepsilon \)} 
\label{sec:countEmptyBorderNonunary}

We now consider the case where \( A \) is a non-unary alphabet and the subword \( w \) under consideration has maximal border \( \varepsilon \), meaning that \( w \) does not overlap itself. As such, once we have started to read \( w \) we can continue reading it until it finishes without worrying that another occurrence of \( w \) may have already begun.

From \eqref{eq:Count0}, we know that the language \( \Count(w, 0) \) can be represented by the star-free expression \( (\emptyset^{c} w \emptyset^{c})^{c} \). Knowing this, we can obtain an expression representing \( \Count(w, k) \) which is star-free:
\begin{align*}
\Count(w, k) = [ \Count(w, 0) \cdot w ]^{k} \cdot \Count(w, 0).
\end{align*}
As can be seen from this expression, we begin with a word from \( \Count(w, 0) \), which may be empty, and then count the \( k \) occurrences of the subword \( w \), with each pair of occurrences `padded' by a word from \( \Count(w, 0) \). We finish with a word from \( \Count(w, 0) \), which, again, may be empty.

We now turn our attention to counting subwords modulo \( n \). An expression for \( \ModCount(w, k, n) \) which is of star-height one is given by
\begin{align*}
[ \Count(w, 0) \cdot w ]^{k} \big[ [ \Count(w, 0) \cdot w ]^{n} \big]^{\ast} \cdot \Count(w, 0).
\end{align*}
As can be seen from this expression, we begin with a word from \( \Count(w, 0) \) and then count the first \( k \) occurrences of the subword \( w \), with each pair of occurrences `padded' by a word from \( \Count(w, 0) \). After this, we allow the same expression to repeat in non-negative multiples of \( n \) before ending with a final word from \( \Count(w, 0) \). 

A combination of the above constitutes a proof for the following lemma:

\begin{lem}
\label{lem:nonunaryEmptyBorderHeights}
Let \( A \) be a non-unary alphabet and let \( w \) have maximal border \( \varepsilon \). Every language \( \Count(w, k) \) is of star-height zero, and every language \( \ModCount(w, k, n) \) is of star-height at most one. \qed
\end{lem}

\subsection{Case 3: a non-unary alphabet and powers of a letter}
\label{sec:countPowerNonunary}

We now analyse the third case where the subword under consideration consists of a power of a letter \( a \) from an alphabet \( A \) which contains at least two letters. Specifically, we are interested in finding generalised regular expressions for the languages \( \Count(a^{r}, k) \) and \( \ModCount(a^{r}, k, n) \), where \( r \) is a natural number.

From \eqref{eq:Count0}, we know that the language \( \Count(a^{r}, 0) \) can be represented by the star-free expression \( (\emptyset^{c} (a^{r}) \emptyset^{c})^{c} \).

In the case where \( k>0 \), we find an expression representing the language \( \Count(a^{r}, k) \) by first considering only those words that have \( a^{r} \) as a border. We denote this language by \( \CountPS(a^{r}, k) \). Let \( B = A \setminus \{a\}\) .
We think of \( B \) as a set of `buffers' that stop us from `accidentally' reading two \( a \)s in a row. This is important as letters may appear as a component of more than one subword and the `buffers' are used to mark the points where we stop reading powers of \( a \). We also define the subset \( W \) of \( A^{\ast} \) by
\begin{align*}
W &= B \cup (B \cdot \Count(a^{r}, 0) \cdot B),
\end{align*}
which is the set of non-empty words that do not feature \( a^{r} \) as a subword and neither start nor end with \( a \). 
It is useful to think of elements of \( W \) as `wedges', separating the strings that feature \( a^{r} \) from one another. Note that the individual components of \( W \) are all star-free expressions which implies that \( W \) is a language of star-height zero.

A general formula for \( \CountPS(a^{r}, k) \) is given by
\begin{align*}
\CountPS(a^{r}, k) =
\bigcup_{j=1}^{k}
\bigcup_{\substack{k_{1}, k_{2}, \dots, k_{j} \geq r \\ k_{1} + k_{2} + \dots + k_{j} = k + \left(r-1\right)j}}
a^{k_{1}} W a^{k_{2}} W \dots W a^{k_{j}},
\end{align*}
where the right-hand side is a regular expression since both unions are finite. Note that the expression is star-free. To see that this equality is correct, consider an arbitrary word \( w \) in \( \CountPS(a^{r}, k) \). Let \( a^{k_{1}}, \dots, a^{k_{j}} \) be the maximal subwords of \( w \) that are powers of \( a \) and have length greater than or equal to \( r \). Note that \( a^{k_{1}} \) must be a prefix of \( w \) as \( w \) starts with \( a^{r} \), and, likewise, \( a^{k_{j}} \) must be a suffix. Hence, we have a decomposition \( w = a^{k_{1}} w_{1} a^{k_{2}} w_{2} \dots w_{j-1} a^{k_{j}} \), where, necessarily, \( w_{1}, \dots, w_{j-1} \) belong to \( W \). Furthermore, each \( a^{k_{i}} \) contains precisely \( k_{i} - r + 1 \) occurrences of \( a^{r} \) by \eqref{eq:Count(a^r,k)}. Since all of the occurrences of \( a^{r} \) appear as subwords of \( a^{k_{i}} \),  we must have
\begin{align*}
k = |w|_{a^{r}} = \sum_{i = 1}^{j} (k_{i} - r + 1) = k_{1} + \dots + k_{j} - (r - 1)j,
\end{align*}
and so \( w \) belongs to the right-hand side. A similar analysis shows that, conversely, every element of the right-side side belongs to \( \CountPS(a^{r}, k) \).

Now, a star-free expression representing the language consisting of \emph{all} words that contain precisely \( k \) occurrences of \( a^{r} \) as a subword, namely \( \Count(a^{r}, k) \), is given by
\begin{align*}
\big[ \varepsilon \cup [ \Count(a^{r}, 0) \cdot B ] \big] \cdot \CountPS(a^{r}, k) \cdot \big[ [ B \cdot \Count(a^{r}, 0) ] \cup \varepsilon \big].
\end{align*}
To see this, note that the \( k \) occurrences of \( a^{r} \) all appear in the central term \( \CountPS(a^{r}, k) \). This term can be preceded by either the empty word or a word that does not contain \( a^{r} \) as a subword; that is, a word from the language \( \Count(a^{r}, 0) \). However, since words in \( \Count(a^{r}, 0) \) have the potential to end with a power of \( a \), we must utilise a `buffer' from the set \( B \). A dual argument deals with potential suffices. Since each of the components of the above expression are star-free, the language \( \Count(a^{r}, k) \) must be of star-height zero.

We now turn our attention to counting occurrences of \( a^{r} \) modulo \( n \). Our strategy here is to count the first \( k \) occurrences of \( a^{r} \) using the expression found above for \( \CountPS(a^{r}, k) \), and then count occurrences of \( a^{r} \) in multiples of \( n \) before adding appropriate prefixes and suffices (as in the case of \( \Count(a^{r}, k) \)).

Having used \( \CountPS(a^{r}, k) \) to count the first \( k \) occurrences of \( a^{r} \), we note that the suffix \( a^{r-1} \) has the potential to be a component of a new occurrence of \( a^{r} \) if the part of the word immediately following \( a^{r-1} \) begins with an \( a \). Similarly, the suffix \( a^{r-2} \) immediately followed by an \( a^{2} \) leads to another occurrence of \( a^{r} \). In order to take these possibilities into account, let \( \Multiple(a^{r}, n) \) denote the language whose words contain precisely \( n \) occurrences of the subword \( a^{r} \) when left concatenated by \( a^{r-1} \) and also have suffix \( a^{r} \):
\begin{align*}
\Multiple(a^{r}, n) = \{ w \in A^{\ast} \mid |a^{r-1}w|_{a^{r}} = n \text{ and } w \text{ has suffix } a^{r} \}.
\end{align*}
The significance of the assumption about the suffix \( a^{r} \) is that every count stops precisely when the \( n \)-th occurrence of \( a^{r} \) is met, and that this suffix `feeds into' the next group of occurrences of \( a^{r} \). 

A star-free expression for this language is given by
\begin{align*}
\Multiple(a^{r}, n) = a^{n} \cup \bigcup_{i = 0}^{n-1} a^{i} W \cdot \CountPS(a^{r}, n-i).
\end{align*}
Note that the right-hand side is a regular expression since the union is finite. To see that this equality is correct, consider an arbitrary word \( w \) in \( \Multiple(a^{r}, n) \). If \( w = a^{k} \) for some natural number \( k \) then
\begin{align*}
n = |a^{r-1}w|_{a^{r}} = |a^{r-1}a^{k}|_{a^{r}} = |a^{r+k-1}|_{a^{r}} = k
\end{align*} 
by \eqref{eq:Count(a^r,k)}, and hence \( w = a^{n} \). Otherwise, we can decompose \( w \) as 
\begin{align*}
w = a^{k_{1}} w_{1} a^{k_{2}} w_{2} \dots w_{j-1} a^{k_{j}},
\end{align*}
where,
\begin{align*}
w_{1}, \dots, w_{j-1} \in W, \qquad k_{1} \geq 0 \qquad \text{and} \qquad  k_{2}, \dots, k_{j} \geq r.
\end{align*}
The maximal subwords of \( a^{r-1}w \) that are powers of \( a \) of exponent greater than or equal to \( r \) are \( a^{r-1}a^{k_{1}} = a^{r+k_{1}-1} \) (provided that \( k_{1} > 0 \)) and \( a^{k_{2}}, \dots, a^{k_{j}} \). Furthermore, our decomposition of \( w \) can be used to split \( a^{r-1}w \) as \( a^{r-1}w = xy \), where \( x = a^{r+k_{1}-1}w_{1} \) and \( y = a^{k_{2}} w_{2} \dots w_{j-1} a^{k_{j}} \). Suppose that \( x \) contains \( i \) occurrences of \( a^{r} \). Then
\begin{align*}
i = |a^{r+k_{1}-1} w_{1}|_{a^{r}} = |a^{r+k_{1}-1}|_{a^{r}} = k_{1}
\end{align*}
by \eqref{eq:Count(a^r,k)}. Moreover, \( y \) must contain the remaining \( n - i \) occurrences of \( a^{r} \) and has \( a^{r} \) as a border. Hence \( y \) belongs to \( \CountPS(a^{r}, n-i) \). Thus, \( w \) belongs to \( a^{i} W \cdot \CountPS(a^{r}, n-i) \) and hence belongs to the union on the right-hand side. A similar analysis shows that, conversely, every element of the right-hand side belongs to \( \Multiple(a^{r}, n) \).

Putting all of this together, we have that an expression representing \( \ModCount(a^{r}, k, n) \), where \( k>0 \), is given by
\[
\big[ \varepsilon \cup [ \Count(a^{r}, 0) \cdot B ] \big] \cdot \CountPS(a^{r}, k) \cdot \Multiple(a^{r}, n)^{\ast} \cdot \big[ [ B \cdot \Count(a^{r}, 0) ] \cup \varepsilon \big],
\]
and an expression representing \( \ModCount(a^{r}, 0, n) \) is given, with slight abuse of notation, by
\begin{align*}
\Count(a^{r}, 0) \cup \ModCount(a^{r}, n, n).
\end{align*}
Both of these expressions are of star-height one, and so the language \( \ModCount(a^{r}, k, n) \) is of star-height at most one.

A combination of the above constitutes a proof for the following lemma:

\begin{lem}
\label{lem:nonunaryPowerHeights}
Let \( A \) be a non-unary alphabet. For every natural number \( r \), the language \( \Count(a^{r}, k) \) is of star-height zero, and the language \( \ModCount(a^{r}, k, n) \) is of star-height at most one. \qed
\end{lem}

\subsection{Discussion} 
\label{sec:wordLengthDiscussion}

Based on the results presented so far, it is natural to ask whether the language \( \Count(w, k) \) is of star-height zero and whether the language \( \ModCount(w, k, n) \) is of star-height at most one for all words \( w \). 
As a consequence of the foregoing results, this is certainly the case for words of length \( \leq 2 \): indeed every such word is either a power of a letter or has maximal border \( \varepsilon \).

\begin{prop}
\label{prop:heightLengthTwoWords}
Let \( A \) be an alphabet. For any word \( w \) in \( A^{+} \) with \( |w| \leq 2 \), the language \( \Count(w, k) \) is of star-height zero, and the language \( \ModCount(w, k, n)\) is of star-height at most one. \qed
\end{prop}

When the subword under consideration is of length three we are presented with a new hurdle to overcome. The possible types for words of length three are
\begin{align*}
aaa, \qquad aab, \qquad aba, \qquad abb, \qquad abc,
\end{align*}
where \( a, b \) and \( c \) are distinct letters in \( A \). Counting occurrences of the word \( aaa \) is covered by Lemmas~\ref{lem:unaryAlphabetHeights} and \ref{lem:nonunaryPowerHeights}, while the words
 \( aab \), \( baa \) and \( abc \) are covered by Lemma~\ref{lem:nonunaryEmptyBorderHeights}. 

With the final type, namely \( aba \), we must be more careful as the maximal border in this case is \( a \), meaning that the suffix \( a \) can act as a prefix \( a \) in a new occurrence of the subword. 
For example, the word \( abababa \) contains three occurrences of the subword \( aba \). 
However, we can proceed in a similar manner to that in Section~\ref{sec:countPowerNonunary} to resolve this issue.

Define \( W \) to be the set of words that are not \( b \), do not have prefix \( ba \), do not have suffix \( ab \), and do not contain \( aba \) as a subword; that is,
\begin{align*}
W = ( b \cup baA^{\ast} \cup A^{\ast}ab \cup A^{\ast}abaA^{\ast} )^{c} = ( b \cup ba\emptyset^{c} \cup \emptyset^{c}ab \cup \emptyset^{c}aba\emptyset^{c} )^{c}.
\end{align*}
Then, a general formula for \( \CountPS(aba, k) \), where \( k \) is a natural number, is given by
\begin{align*}
\CountPS(aba, k) =
\bigcup_{j=1}^{k}
\bigcup_{\substack{k_{1}, k_{2}, \dots, k_{j} \geq 1 \\ k_{1} + k_{2} + \dots + k_{j} = k}}
a(ba)^{k_{1}} W a(ba)^{k_{2}} W \dots W a(ba)^{k_{j}},
\end{align*}
which is star-free, and the language \( \Count(aba, k) \), expressed by
\begin{align*}
(\emptyset^{c}aba\emptyset^{c} \cup \emptyset^{c}ab)^{c} \cdot \CountPS(aba, k) \cdot (ba\emptyset^{c} \cup \emptyset^{c}aba\emptyset^{c})^{c},
\end{align*}
is of star-height zero.

To find an expression for \( \ModCount(aba, k, n) \) we introduce the language
\begin{align*}
\Multiple(aba, n) = \{ w \in A^{\ast} \mid |aw|_{aba} = n \text{ and } w \text{ has suffix } aba \}.
\end{align*}
A star-free expression representing \( \Multiple(aba, n) \) is given by
\begin{align*}
(ba)^{n} \cup \bigcup_{i = 1}^{n-1} (ba)^{i} W \cdot \CountPS(aba, n-i).
\end{align*}
Putting all of this together, an expression representing \( \ModCount(aba, k, n) \), where \( k > 0 \), is given by
\begin{align*}
(\emptyset^{c}aba\emptyset^{c} \cup \emptyset^{c}ab)^{c} \cdot \CountPS(aba, k) \cdot \Multiple(aba, n)^{\ast} \cdot (ba\emptyset^{c} \cup \emptyset^{c}aba\emptyset^{c})^{c},
\end{align*}
and an expression representing \( \ModCount(aba, 0, n) \) is given, with slight abuse of notation, by
\begin{align*}
\Count(aba, 0) \cup \ModCount(aba, n, n).
\end{align*}
This establishes that the language \( \ModCount(aba, k, n) \) is of star-height at most one.

Hence, we have proven the following result:

\begin{prop}
\label{prop:heightLengthThreeWords}
Let \( A \) be an alphabet. For any word \( w \) in \( A^{+} \) with \( |w| \leq 3 \), the language \( \Count(w, k) \) is of star-height zero, and the language \( \ModCount(w, k, n)\) is of star-height at most one. \qed
\end{prop}

It should be noted that Proposition~\ref{prop:heightLengthThreeWords} can also be proved using existing theoretical results. We briefly outline the proof strategy below.

Let \( A \) and \( X = \{x\} \) be alphabets, and consider the languages \( L = \ModCount(x, k, n) = x^{k}(x^{n})^{\ast} \) over \( X \) and \( K = \ModCount(w, k, n) \) over \( A \). Define a function \( f_{w} : A^{\ast} \to X^{\ast} \) by \( f_w(v) = x^{|v|_{w}} \). It is easy to show that \( K = Lf_{w}^{-1} \).

Note that for all words \(w \) with \( |w| \leq 3 \), \( f_{w} \) is a generalised sequential function (in the sense of Eilenberg \cite[p.~299]{eilenberg74}). For example, a transducer realising \( f_{aba} \) is shown in Figure~\ref{fig:transducer}. In this diagram, edges labelled with \( c|\varepsilon \), where \( c \in A \setminus \{a,b\} \), have been removed for clarity, since all of these edges point directly to the initial state.

Standard calculations show that the transition monoid of each transducer realising \( f_{w} \), where \( |w| \leq 3 \) is aperiodic. Moreover, the transition monoid of the automaton recognising \( L \) is an abelian group by \ref{lem:commGrpIffBool} and \ref{prop:commGrpHeightOne}. Hence, by Eilenberg \cite[Proposition IX.1.1]{eilenberg76} (suitably modified to deal with generalised sequential functions), the transition monoid of \( K \) divides a wreath product of an abelian group by an aperiodic monoid. Since all languages that belong to the pseudovariety generated by wreath products of abelian groups by aperiodic monoids have star-height at most one \cite[Theorem 7.8]{pin92}, we conclude that \( K \) is of star-height at most one.

\begin{figure}
\begin{center}
\includegraphics{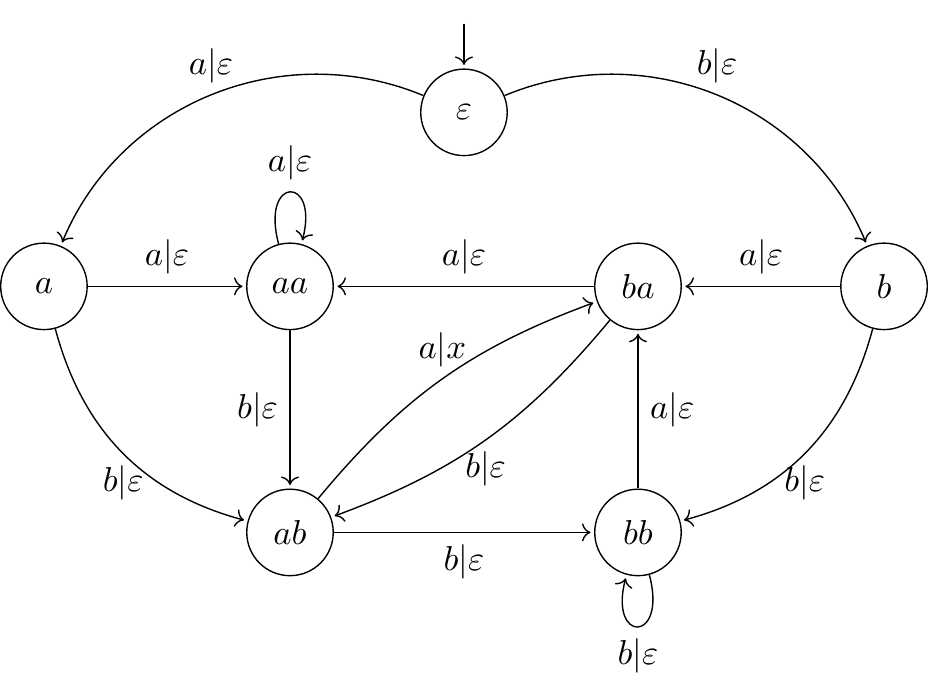}
\end{center}
\caption{A finite state transducer realising \( f_{aba} \).}
\label{fig:transducer}
\end{figure}

\section{Applications to Rees zero-matrix semigroups}
\label{sec:ReesMatrix}

In this section, we change tack and use our combinatorial results to prove new results in an algebraic setting. Specifically, we show that languages recognised by Rees zero-matrix semigroups over abelian groups are of star-height at most one.

Let \( S \) be a semigroup without zero. Let \( I \) and \( \Lambda \) be non-empty indexing sets and let \( P \) be a \( \left\vert{\Lambda}\right\vert \times \left\vert{I}\right\vert \) matrix with entries from \( S \cup \{\mathbf{0}\} \), where \( \mathbf{0} \) is a new symbol not in \( S \). The \emph{Rees zero-matrix semigroup} \( M^{0}[S; I, \Lambda; P] \) is the set \( (I \times S \times \Lambda) \cup \{\mathbf{0}\} \) equipped with the binary operation defined by
\begin{align*}
(i, s, \lambda)(j, t, \mu) =
\begin{cases}
(i, sp_{\lambda j}t, \mu) & \text{if } p_{\lambda j} \neq \mathbf{0}, \\
\mathbf{0} & \text{if } p_{\lambda j} = \mathbf{0},
\end{cases}
\end{align*}
and \( s\mathbf{0} = \mathbf{0} = \mathbf{0}s \) for all \( s \) in \( S \cup \{\mathbf{0}\} \). If we disregard the new symbol \( \mathbf{0} \) but leave everything else intact then the resulting semigroup, denoted by \( M[S; I, \Lambda; P] \), is simply a \emph{Rees matrix semigroup}. Throughout the rest of this section we work in full generality with Rees zero-matrix semigroups. The results of course remain true when restricted to Rees matrix semigroups.

We say that the matrix $P$ is \emph{regular} if each row and column contain a non-zero entry.
Rees zero-matrix semigroups with finite underlying groups and regular matrices are precisely finite $0$-simple semigroups according to Rees' Theorem \cite[Theorem 3.2.3]{howie95}. In turn, these semigroups together with zero semigroups completely exhaust principal factors of arbitrary finite semigroups.

We begin by exploring which languages are recognised by Rees zero-matrix semigroups over cyclic groups \( \mathbb{Z}_{n} \), where \( n \) is a natural number, and then extend this to arbitrary abelian groups via the Fundamental Theorem for Finite Abelian Groups.

\subsection{Rees zero-matrix semigroups over cyclic groups}
\label{sec:ReesMatrixCyclic}

Let \( S = M^{0}[\mathbb{Z}_{n}; I, \Lambda; P] \) be a Rees zero-matrix semigroup, where the zero in \( S \) is denoted by \( \mathbf{0} \) and the identity in \( \mathbb{Z}_{n} \) is denoted by \( 0 \). Let \( A \) be an alphabet and define a map \( \varphi : A \to S \) by either \( a\varphi = \mathbf{0} \) or \( a\varphi = (i_{a}, g_{a}, \lambda_{a}) \), where \( 0 \leq g_{a} < n \). Let 
\[ 
A_{(i, g, \lambda)} =(i,g,\lambda)\varphi^{-1} \qquad \text{and} \qquad  A_{\mathbf{0}}=\mathbf{0}\varphi^{-1}.
\] 
Uniquely extend \( \varphi \) to a morphism \( \bar{\varphi} : A^{+} \to S \). 

Now, consider the image of $w=a_1a_2\dots a_r$ under $\bar{\varphi}$.
If \( a_{t}\bar{\varphi}=\mathbf{0} \) for at least one \( t \in \{1, 2, \dots, r\} \) then 
\( w \bar{\varphi} = \mathbf{0} \). Likewise, if  \( p_{\lambda_{a_{t}} i_{a_{t+1}}} = \mathbf{0} \) for at least one \( t \in \{1, 2, \dots, r-1\} \) then \( w \bar{\varphi} = \mathbf{0} \). 
Otherwise, if \( a_{t} \varphi \neq \mathbf{0} \) for all \( t \in \{1, 2, \dots, r\} \) and \( p_{\lambda_{a_{t}} i_{a_{t+1}}} \neq \mathbf{0} \) for all \( t \in \{1, 2, \dots, r-1\} \), then
\begin{align*}
w\bar{\varphi} &=
(i_{a_{1}}, g_{a_{1}}, \lambda_{a_{1}})(i_{a_{2}}, g_{a_{2}}, \lambda_{a_{2}}) \dots (i_{a_{r}}, g_{a_{r}}, \lambda_{a_{r}}) \\
&= (i_{a_{1}}, g_{a_{1}} + p_{\lambda_{a_{1}}i_{a_{2}}} + g_{a_{2}} + p_{\lambda_{a_{2}}i_{a_{3}}} + \dots + p_{\lambda_{a_{r-1}}i_{a_{r}}} + g_{a_{r}}, \lambda_{a_{r}}).
\end{align*}

We proceed by finding regular expressions for preimages of elements in \( S \). We split into two cases: the preimage of the zero \( \mathbf{0} \) and the preimage of an arbitrary non-zero element \( s = (i, g, \lambda) \).

\begin{lem}
\label{lem:preimageOfZero}
With the notation as above, \( \mathbf{0}\bar{\varphi}^{-1} \) is of star-height zero.
\end{lem}

\begin{proof}
According to the analysis preceding the lemma, a word \( w = a_1a_2\dots a_r \) belongs to the preimage of \( \mathbf{0} \) if and only if at least one of the following holds:
\begin{enumerate}
\item $a_t$ lies in $A_\mathbf{0}$ for some $t$ in $\{1,2\dots,r\}$; or
\item $p_{\lambda j}=\mathbf{0}$, where $a_t\in A_{(i,g,\lambda)}$ and $a_{t+1}\in A_{(j,h,\mu)}$.
\end{enumerate}
It follows that
\[
\mathbf{0}\bar{\varphi}^{-1}=
A^{\ast}A_{\mathbf{0}}A^{\ast} \cup \Bigl[\bigcup A^{\ast}A_{(i, g, \lambda)}A_{(j, h, \mu)}A^{\ast}\Bigr],
\]
where the second union is taken over all $(i,g,\lambda),(j,h,\mu)\in S\setminus\{\mathbf{0}\}$ with 
$p_{\lambda j}=\mathbf{0}$,
a language of star-height zero by Observation \ref{obs:subsetOfAlphabet}.
\end{proof}

\begin{lem}
For a non-zero element $s=(i,g,\lambda)$ in $S$,  its preimage, \( s\bar{\varphi}^{-1} \), is of star-height at most one.
\end{lem}

\begin{proof}
We begin by writing \( s\bar{\varphi}^{-1} \) as the intersection of three regular languages as follows:
\begin{align}
\label{eq:sPreimageDecomp}
s\bar{\varphi}^{-1} = \left( \{i\} \times \mathbb{Z}_{n} \times \Lambda \right)\bar{\varphi}^{-1} \cap \left( I \times \{g\} \times \Lambda \right)\bar{\varphi}^{-1} \cap \left( I \times \mathbb{Z}_{n} \times \{\lambda\} \right)\bar{\varphi}^{-1}.
\end{align}

Due to the nature of the multiplication on \( S \), it is clear to see that
\begin{eqnarray*}
&&
\left( \{i\} \times \mathbb{Z}_{n} \times \Lambda \right)\bar{\varphi}^{-1} = \bigg[ \bigcup_{h \in \mathbb{Z}_{n}, \mu \in \Lambda} A_{(i, h, \mu)} \bigg] \cdot A^{\ast},
\\
&&\left( I \times \mathbb{Z}_{n} \times \{\lambda\} \right)\bar{\varphi}^{-1} = A^{\ast} \cdot \bigg[ \bigcup_{j \in I, h \in \mathbb{Z}_{n}} A_{(j, h, \lambda)} \bigg].
\end{eqnarray*}
By Observation~\ref{obs:subsetOfAlphabet}, these languages have star-height zero.

It remains to find an expression for 
\( \left( I \times \{g\} \times \Lambda \right)\bar{\varphi}^{-1} \). 
Consider an arbitrary \( w=a_1a_2\dots a_r \) belonging to this language.
Continuing to use the notation introduced before Lemma \ref{lem:preimageOfZero},
we know that \( p_{\lambda_{a_t} i_{a_{t+1}}}\neq \mathbf{0} \) for \( t = 1, 2, \dots, r-1 \), and
\begin{align*}
g_{a_{1}} + p_{\lambda_{a_{1}}i_{a_{2}}} + g_{a_{2}} + p_{\lambda_{a_{2}}i_{a_{3}}} + \dots + p_{\lambda_{a_{r-1}}i_{a_{r}}} + g_{a_{r}} &\equiv g \pmod n.
\end{align*}
We split the above sum into two:
\begin{align*}
\underbrace{g_{a_{1}} + g_{a_{2}} + \dots + g_{a_{r}}}_{ {} \equiv g_{1} \pmod n} + \underbrace{p_{\lambda_{a_{1}}i_{a_{2}}} + p_{\lambda_{a_{2}}i_{a_{3}}} + \dots + p_{\lambda_{a_{r-1}}i_{a_{r}}}}_{ {} \equiv g_{2} \pmod n} &\equiv g \pmod n,
\end{align*}
and we examine them separately. The first sum corresponds to the contributions from `group' summands, while the second is the contributions from `matrix' summands.

For the group contribution, we consider the congruence given by
\[
g_{a_{1}} + g_{a_{2}} + \dots + g_{a_{r}} \equiv g_{1} \pmod n.
\]
Grouping together summands corresponding to the same letter, we see that the above congruence is equivalent to
\[
\sum_{a \in A} g_{a} |w|_{a} \equiv g_{1} \pmod n,
\]
which, in turn, is equivalent to
\[
\sum_{a \in A} g_{a} (|w|_{a} \pmod n) \equiv g_{1} \pmod n.
\]
The point here is that while \( |w|_a \) can take infinitely many values, the same is not true for \( |w|_a\pmod{n} \). More formally, let  \( T \) be the following set of tuples of elements \( \{0, 1, \dots, n-1\} \) indexed by \( A \):
\[
T = \{(k_{a})_{a \in A} \mid \sum_{a \in A} g_{a}k_{a} \equiv g_{1} \pmod n\}.
\]
For any fixed tuple \( (k_{a})_{a \in A} \) in \( T \), every word \( w \) such that \( |w|_a\equiv k_a\pmod{n}\), where \( a \) lies in \( A \), will have group contribution equal to \( g_{1} \bmod n \).
The set of all such words is obtained
by forming the finite intersection  of the languages \( \ModCount(a, k_{a}, n) \) for \( a \in A \). 
Taking the finite union over all tuples in \( T \) results in the expression
\[
\GrpContribution(g_{1}, n) =
\bigcup_{(k_{a}) \in T} \bigcap_{a \in A} \ModCount(a, k_{a}, n),
\]
which is of star-height at most one, since \( \ModCount(a, k_{a}, n) \) is of star-height at most one by Lemmas~\ref{lem:unaryAlphabetHeights} and \ref{lem:nonunaryEmptyBorderHeights}.

In a similar fashion, we consider the contributions made by `matrix' summands; that is, we consider the congruence given by
\[
p_{\lambda_{a_{1}}i_{a_{2}}} + p_{\lambda_{a_{2}}i_{a_{3}}} + \dots + p_{\lambda_{a_{r-1}}i_{a_{r}}} \equiv g_{2} \pmod n.
\]
Counting the contribution of each matrix entry separately, we see that the above congruence is equivalent to
\[
\sum_{ab \in A^{2}} p_{\lambda_{a} i_{b}} |w|_{ab} \equiv g_{2} \pmod n,
\]
which, in turn, is equivalent to
\[
\sum_{ab \in A^{2}} p_{\lambda_{a} i_{b}} (|w|_{ab} \pmod n) \equiv g_{2} \pmod n.
\]
Consider the finite family $U$ of tuples \( (k_{ab})_{ab \in A^{2}} \) of elements\( \{0, 1, \dots, n-1 \}\), indexed by $A^2$:
\[
U = \{(k_{ab})_{ab \in A^{2}} \mid \sum_{ab \in A^{2}} p_{\lambda_{a}i_{b}}k_{ab} \equiv g_{2} \pmod n\}.
\]
For a fixed tuple in \( U \), the set of all words $w$ satisfying $|w|_{ab}\equiv k_{ab}\pmod{n}$, where $ab$ lies in $A^2$,
is obtained by taking the finite intersection of the languages \( \ModCount(ab, k_{ab}, n) \). 
Taking the union over all tuples in \( U \) yields
\[
\MatContribution(g_{2}, n) =
\bigcup_{(k_{ab})_{ab \in A^{2}} \in U} \bigcap_{ab \in A^{2}} \ModCount(ab, k_{ab}, n),
\]
which is of star-height at most one, since \( \ModCount(ab, k_{ab}, n) \) is of star-height at most one by Proposition~\ref{prop:heightLengthTwoWords}.

Combining the `group' contribution and the `matrix' contribution appropriately leads to
\[
\left( I \times \{g\} \times \Lambda \right)\bar{\varphi}^{-1} =
\bigcup_{\substack{(g_{1}, g_{2}) \in \mathbb{Z}_{n}^{2} \\ g_{1} + g_{2} \equiv g \pmod n}}  (\GrpContribution(g_{1}, n) \cap \MatContribution(g_2,n)),
\]
and completes the proof.
\end{proof}

An immediate consequence of the above propositions is the following theorem:

\begin{thm}
\label{thm:RMSOverCyclicHeightOne}
A regular language recognised by a Rees zero-matrix semigroup over a cyclic group is of star-height at most one.
\end{thm}

\begin{proof}
Every language recognised by a Rees zero-matrix semigroup over a cyclic group can be expressed as a finite union of preimages of elements in the semigroup. Since each individual preimage is of star-height at most one and taking finite unions does not increase star-height, the result follows.
\end{proof}

\subsection{Extending to abelian groups}
\label{sec:ReesMatrixComm}

We now extend Theorem~\ref{thm:RMSOverCyclicHeightOne} to Rees zero-matrix semigroups over abelian groups. In order to do this we make use of properties of homomorphisms and projection maps and appeal to the Fundamental Theorem of Finite Abelian Groups.

We begin with some general theory concerning Rees matrix semigroups over direct products of semigroups.
Consider a Rees zero-matrix semigroup $M^0[S\times T; I,\Lambda; R]$, with $R=(r_{\lambda i})$,
where $r_{\lambda i}=(p_{\lambda i},q_{\lambda i})$ lies in $S\times T$ or $r_{\lambda i}=\mathbf{0}_{S\times T}$, the zero element.
Define two further Rees matrix semigroups \( M^{0}[S; I, \Lambda; P] \) and \( M^{0}[T; I, \Lambda; Q] \),
with zeros $\mathbf{0}_S$ and $\mathbf{0}_T$ respectively, and matrices $P$ and $Q$ defined
by $P=(p_{\lambda i})$ and $Q=(q_{\lambda i})$, 
where we take $p_{\lambda i}=\mathbf{0}_S$ and $q_{\lambda i}=\mathbf{0}_T$ whenever $r_{\lambda i}=\mathbf{0}_{S\times T}$.
We then have two natural projections:
\begin{eqnarray*}
&&\pi_S \::\: M^0[S \times T; I,\Lambda; R]\rightarrow M^{0}[S; I, \Lambda; P] \::\: (i,(s,t),\lambda)\mapsto (i,s,\lambda),
\\
&&\pi_T \::\: M^0[S \times T; I,\Lambda; R]\rightarrow M^{0}[T; I, \Lambda; Q] \::\: (i,(s,t),\lambda)\mapsto (i,t,\lambda).
\end{eqnarray*}
Proof that these are epimorphisms is routine and is left as an exercise.

Now suppose that we are given an alphabet $A$ and a map $\varphi\::\: A\rightarrow M^0[S\times T; I,\Lambda; R]$, which extends uniquely to a homomorphism $\bar{\varphi}\::\: A^+\rightarrow M^0[S\times T; I,\Lambda; R]$. Then the compositions $\bar{\varphi}\pi_S$ and $\bar{\varphi}\pi_T$ are homomorphisms from $A^+$ to \( M^{0}[S; I, \Lambda; P] \) and \( M^{0}[T; I, \Lambda; Q] \) respectively. The entire set-up is summarised in the following diagram:

\begin{center}
\includegraphics{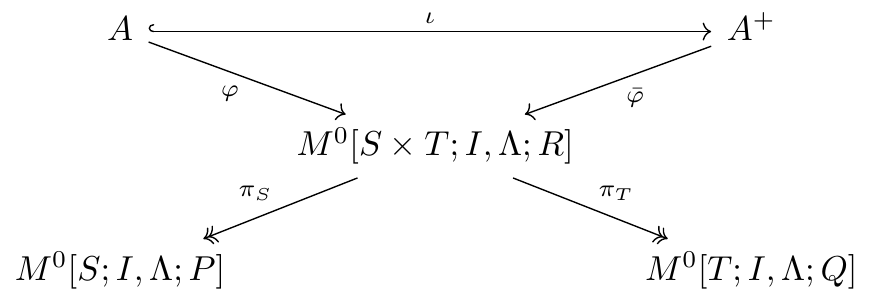}
\end{center}

In the following lemma we relate the preimage of a non-zero element in \( M^{0}[S \times T; I, \Lambda; R] \) to the preimages of non-zero elements in \( M^{0}[S; I, \Lambda; P] \) and \( M^{0}[T; I, \Lambda; Q] \).

\begin{lem}
\label{lem:preimageExtension}
For any \( \left( i, (s, t), \lambda \right) \) in \(  M^{0}[S \times T; I, \Lambda; R] \) we have
\begin{align*}
\left( i, (s, t), \lambda \right) \bar{\varphi}^{-1} &= (i, s, \lambda)(\bar{\varphi}\pi_{S})^{-1} \cap (i, t, \lambda)(\bar{\varphi}\pi_{T})^{-1},
\\
\mathbf{0}_{S\times T} \bar{\varphi}^{-1} &= \mathbf{0}_{S}(\bar{\varphi}\pi_{S})^{-1} \cap \mathbf{0}_{T}(\bar{\varphi}\pi_{T})^{-1}.
\end{align*}
\end{lem}

\begin{proof}
First, suppose that \( w \in \left( i, (s, t), \lambda \right) \bar{\varphi}^{-1} \); that is, \( w\bar{\varphi} = \left( i, (s, t), \lambda \right) \). Then
\[
w(\bar{\varphi}\pi_{S}) = \left( i, (s, t), \lambda \right) \pi_{S} = (i, s, \lambda)
\]
Hence \( w \in (i, s, \lambda)(\bar{\varphi}\pi_{S})^{-1} \), and, analogously, \( w \in (i, t, \lambda)(\bar{\varphi}\pi_{T})^{-1} \). 

Conversely, suppose that \( w \in (i, s, \lambda)(\bar{\varphi}\pi_{S})^{-1} \cap (i, t, \lambda)(\bar{\varphi}\pi_{T})^{-1} \),
so that  \( w(\bar{\varphi}\pi_{S}) = (i, s, \lambda) \) and \( w(\bar{\varphi}\pi_{T}) =  (i, t, \lambda) \). 
Note that $w\bar{\varphi}\neq \mathbf{0}_{S\times T}$, so we must have that
\( w\bar{\varphi} = \left(i_{w}, (s_{w}, t_{w}), \lambda_{w} \right) \) for some \( i_{w} \in I \), \( (s_{w}, t_{w}) \in S \times T \) and \( \lambda_{w} \in \Lambda \). Now,
\[
(i, s, \lambda) = (w\bar{\varphi})\pi_{S} = \left(i_{w}, (s_{w}, t_{w}), \lambda_{w} \right)\pi_{S} = (i_{w}, s_{w}, \lambda_{w} )
\]
and, similarly, \( (i, t, \lambda) = (i_{w}, t_{w}, \lambda_{w}) \). Hence, \( i_{w} = i \), \( s_{w} = s \), \( t_{w} = t \) and \( \lambda_{w} = \lambda \). Therefore, \( w\bar{\varphi} = \left(i_{w}, (s_{w}, t_{w}), \lambda_{w} \right) = \left(i, (s, t), \lambda \right) \) and \( w \in \left(i, (s, t), \lambda \right)\bar{\varphi}^{-1} \), as required.

The second equality is proved in essentially the same way.
\end{proof}

We can now prove the following:

\begin{thm}
\label{cor:directProdKeepsHeight}
Let \( S \) and \( T \) be finite semigroups. If languages recognised by finite Rees zero-matrix semigroups over $S$ or $T$ all have star-height $\leq h$, then all the languages recognised by  the finite Rees zero-matrix semigroups over the direct product $S\times T$ also have star-height $\leq h$.
\end{thm}

\begin{proof}
Lemma \ref{lem:preimageExtension} allows us to express the preimage of an element in the Rees zero-matrix semigroup over the direct product as the intersection of two preimages of elements in Rees zero-matrix semigroups over the factors.
Since the preimage of any subset is a finite union of preimages of elements, the result follows.
\end{proof}

By combining the above results we can now extend Theorem~\ref{thm:RMSOverCyclicHeightOne} to Rees zero-matrix semigroups over abelian groups.

\begin{thm}
\label{thm:RZMSAbelianHeightOne}
A regular language recognised by a Rees zero-matrix semigroup over an abelian group is of star-height at most one.
\end{thm}

\begin{proof}
Invoking the Fundamental Theorem of Finite Abelian Groups and applying Corollary~\ref{cor:directProdKeepsHeight} a finite number of times to Rees zero-matrix semigroups over cyclic groups yields the result.
\end{proof}

Theorem~\ref{thm:RZMSAbelianHeightOne} can also be deduced from existing theoretical results when attention is restricted to the basic Rees matrix construction (without zero). Indeed, let \( S \) be a Rees matrix semigroup over an abelian group \( G \). By  \cite[Proposition XI.3.1]{eilenberg76}, \( S \) divides a wreath product of \( G \) by an aperiodic monoid. However, by \cite[Theorem 7.8]{pin92}, every language recognised by \( \mathbf{Gcom \ast A} \) is of star-height at most one, where \( \mathbf{Gcom \ast A}\) is the pseudovariety generated by wreath products of commutative groups by aperiodic monoids. Hence, \( S \) belongs to the pseudovariety \( \mathbf{Gcom \ast A} \) and every language recognised by \( S \) is of star-height at most one.
\medskip

\noindent
\textbf{Acknowledgement.}
The authors would like to thank the anonymous referee for their suggestions concerning alternative proof strategies for some of the results.


\begin{thebibliography}{1}

\bibitem{eggan63}
L.~C. Eggan.
\newblock Transition graphs and the star-height of regular events.
\newblock {\em Michigan Math. J.}, 10:385--397, 1963.

\bibitem{eilenberg74}
S.~Eilenberg.
\newblock{\em Automata, Languages and Machines; Volume A}.
\newblock Academic Press, 1974.

\bibitem{eilenberg76}
S.~Eilenberg.
\newblock{\em Automata, Languages and Machines; Volume B}.
\newblock Academic Press, 1976.

\bibitem{hashiguchi83}
K.~Hashiguchi.
\newblock Representation theorems on regular languages.
\newblock {\em J. Comput. System Sci.}, 27:101--115, 1983.

\bibitem{henneman71}
W.~H. Henneman.
\newblock {\em Algebraic theory of automata}.
\newblock PhD thesis, MIT, 1971.

\bibitem{howie95}
J.M. Howie.
\newblock {\em Fundamentals of semigroup theory}, volume~12 of {\em London
  Mathematical Society Monographs. New Series}.
\newblock The Clarendon Press, Oxford University Press, New York, 1995.
\newblock Oxford Science Publications.

\bibitem{pin86}
J.-E. Pin.
\newblock {\em Varieties of Formal Languages}.
\newblock North Oxford Academic, 1986.

\bibitem{pin89}
J.-E. Pin, H.~Straubing, and D.~Therien.
\newblock New results on the generalized star-height problem.
\newblock {\em STACS 89, Lecture Notes in Computer Science}, 349:458--467,
  1989.

\bibitem{pin92}
J.-E. Pin, H.~Straubing, and D.~Th{\'e}rien.
\newblock Some results on the generalized star-height problem.
\newblock {\em Inform. and Comput.}, 101(2):219--250, 1992.

\bibitem{sak09}
J.~Sakarovitch.
\newblock {\em Elements of automata theory}.
\newblock Cambridge University Press, Cambridge, 2009.


\end{thebibliography}
\end{document}